\documentclass{aims}
\usepackage{amsmath}
  \usepackage{paralist}
  \usepackage{graphics} %% add this and next lines if pictures should be in esp format
  \usepackage{epsfig} %For pictures: screened artwork should be set up with an 85 or 100 line screen
\usepackage{graphicx}  \usepackage{epstopdf}%This is to transfer .eps figure to .pdf figure; please compile your paper using PDFLeTex or PDFTeXify. 
 \usepackage[colorlinks=true]{hyperref}
\hypersetup{urlcolor=blue, citecolor=red}

\newtheorem{theorem}{Theorem}[section]

\newtheorem{lemma}[theorem]{Lemma}

\theoremstyle{definition}

\newcommand{\be}[1]{\begin{equation}\label{#1}}
\newcommand{\ee}{\end{equation}}

\title[Tolman--Oppenheimer--Volkoff Equation]
      {Tolman--Oppenheimer--Volkoff Equation}

\author[Dorota Bors and Robert Sta\'nczy]{}

 
\begin{document}

%The abstract of your paper

\maketitle

\centerline{\scshape Robert Sta\'nczy}
\medskip
{\footnotesize
 % please put the address of the second  and third author
 \centerline{ Instytut Matematyczny}
 \centerline{Uniwersytet Wroc{\l}awski}
 \centerline{pl. Grunwaldzki 2/4}
 \centerline{50--384 Wroc{\l}aw, Poland}
}

\medskip

% Enter the first author's name and address:
\centerline{\scshape Dorota Bors}
\medskip
{\footnotesize
% please put the address of the first author
 \centerline{Faculty of Mathematics and Computer Science}
   \centerline{University of Lodz}
   \centerline{Banacha 22}
   \centerline{90-238 \L\'od\'z, Poland}
} % Do not forget to end the {\footnotesize by the sign }

\bigskip

% The name of the associate editor will be entered by an editorial staff
% "Communicated by the associate editor name" is not needed for special issue.
% \centerline{(Communicated by the associate editor name)}

\bigskip

\begin{abstract}
The existence of solutions to Tolman--Openheimer--Volkoff equation with linear equation of state modeling relativistic cloud of interacting particles is proved for mass parameter below certain threshold. For the intermediate values of mass parameters multiplicity result holds. For the small mass parameter the uniqueness is guaranteed. Moreover, there is a threshold value of the mass, which can not be exceeded. It is achieved by considering the related dynamical system in the rescaled mass--density variables which is governed by the global Lyapunov function with sink. Some preliminary extensions to nonlinear equation of state are discussed with numerical density profiles.
\end{abstract}

\section{Introduction}
While `black hole' model was suggested by Schwarzschild swiftly after publication of Einstein  equations, Einstein himself was dubious that such aggregation can occur in real life expressing it in 1939 in \cite{Ein}, while at the same time Tolman \cite{Tol} and Oppenheimer with Volkoff \cite{OV} considered the alternative solutions to the problem, extending it with Snyder \cite{OS} while Buchdahl \cite{Buc} and Bondi \cite{B2, B1} showed its limits. We consider systems governed by the Einstein--Euler equations equipped with Equation of State, relating the pressure to the density of energy. We are interested in static and spherically symmetric solutions to the problem. The main results hold for the linear Equation of State motivated by the research of Chandrasekhar \cite{Ch} and more recently Chavanis \cite{a,b,e} and Andersson et all, see \cite{An}. In the more physically accepted nonlinear Equation of State like this modeled by relativistic Fermi--Dirac or relativistic Michie--King that could be modeled by toy polytropic models the linear approximation holds only at large values of the energy density or pressure. These distribution functions were motivated by Ruffini et all \cite{Sag} to model dark matter distribution to explain the recent observation of trajectories of the  objects in our Milky Way with stable \cite{BSM} or unstable Hills effect, see Koposov \cite{S5K}. Our approach relies on change of variables in the spirit of Milne and rigorous proof of the existence of the Lyapunov function thus extending the approach to relativistic from Newtonian approach motivated by the papers of Biler, Hilhorst, Nadzieja \cite{BHN} and \cite{BS, BSJ} for different distribution functions. This allows us to prove the critical mass theorem with analytic formula for the threshold mass while some numerical experiments are also conducted to indicate possible limits. In the last sections we thoroughly examine the relativistic Fermi--Dirac and Michie--King distribution functions and resulting Equations of State together with its asymptotics and interpretations. For these models numerics is available yielding the trajectory along unstable manifold and we depict the corresponding density profiles.

\section{Derivation of the dynamical systems in the spirit of Milne variables}

Consider the Tolman--Oppenheimer--Volkoff equation with the linear $p=\rho$ pressure--density relation
\be{TOV}
-rp'(r-2m)=(p+\rho)(m+4\pi r^3 p)\,
\ee
where $m=\int_0^r s^2\rho(s)ds$. Introducing the new variables interpreted vaguely as the rescaled mass--energy and the density of the system with $r=e^s$ as
\be{Syst:Auton}\left\{\begin{array}{l}
4\pi rx(\ln r )=m(r)\,,\\[6pt]
4\pi y(\ln r)=m'(r)=r^2\rho(r)\,.
\end{array}\right.\ee
one obtains the system
\be{Syst:Nonau}\left\{\begin{array}{l}
x'(s)=-\,x(s)+\,y(s)\,,\\[6pt]
y'(s)=2\,y(s)-\frac{8\pi y(s)(x(s)+y(s))}{1-8\pi x(s)}\,,
\end{array}\right.\ee
For this system (\ref{Syst:Nonau}) we shall show that the point $(0,0)$ is a saddle, with nonstable manifold going out tangentially to the vector $(3,1)$ while the other stationary point $(1/{16 \pi},1/{16\pi})$ is a sink which will be shown below. 

The most important feature and novelty is the existence of the global Lyapunov function
$$
L(x,y)=2+16\pi (y-3x)-\log(128\pi y(1-8\pi x)^3)\,
$$
governs convergence towards this point as was established in similar nonrelativistic model \cite{DSD} and started in \cite{MR0340701}. Indeed multiplying 
the equations (\ref{Syst:Auton}) for $x'$ by $x-2$ and $y'/y$ by $2(2-d)$ and summing them with added $y'$ one obtains
$$
\frac{d}{dt}L(x(t),y(t))=x'(t)(x(t)-2)+y'(t)-2(d-2)y'(t)/y(t)=-(x(t)-2)^2\le 0\,.
$$
Moreover, using Taylor expansion for $(x,y)\sim (\frac{1}{16\pi},\frac{1}{16\pi})$ we can see that
$$L(x,y)\sim \frac12 \left(x-\frac{1}{16\pi}\right)^2+\left(y-\frac{1}{16\pi}\right)^2\,.$$

\begin{figure}[ht]
\label{Fig:Lap}
\begin{center}
\includegraphics[height=5.5cm]{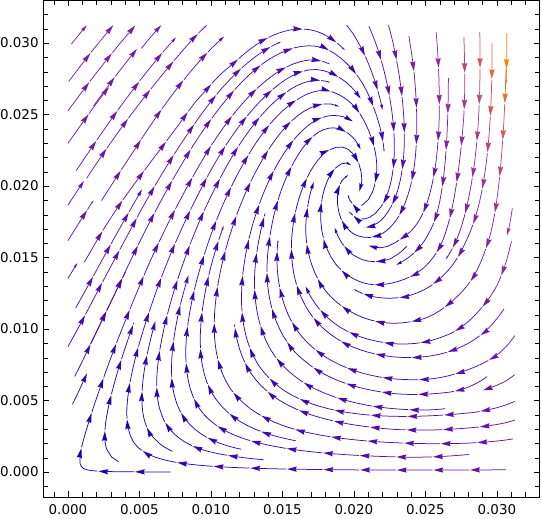}\hspace*{30pt}\includegraphics[height=5.5cm]{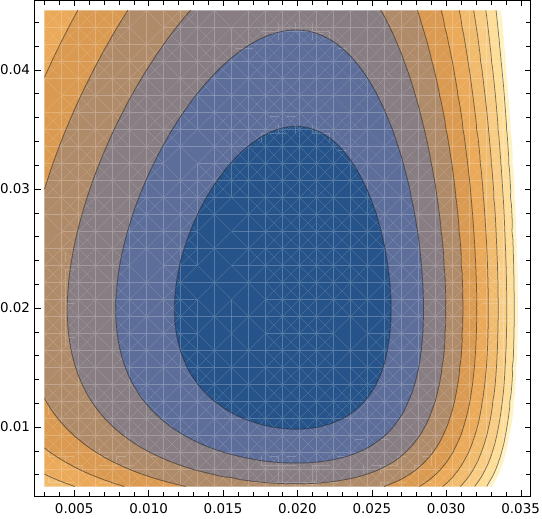}
\end{center}
\caption{\it\small Left: the sink for the flow at $(\frac{1}{16\pi},\frac{1}{16\pi})$ governed by the Lyapunov function $L$;
 Right: the level sets of the Lyapunov function $L$ }
\end{figure}

Furthermore, note that
$$
\lim_{s\rightarrow -\infty} x(s)e^{s} = 0\,,
$$
while assuming $\rho\in L^{\infty}$ guarantees $x(s)e^{-2s}$ to be bounded. Moreover, if $\rho$ is continuous then the following limit exists and is finite
$$
\lim_{s\rightarrow -\infty} x(s)e^{-2s}<\infty\,.
$$
Additionally,
$$
\rho_0=\rho(0)=|\rho|_\infty=\lim_{s\rightarrow -\infty} y(s)e^{-2s}<\infty\,.
$$
Indeed, one can see that using appropriate substitutions one can obtain
$$
\lim_{s\rightarrow -\infty} x(s)e^{-2s}=\lim_{s\rightarrow -\infty} Q'(s)e^{(1-d)s}=\lim_{r\rightarrow 0^+} r^{d-1}\rho(r) r^{1-d}=\rho(0)\,.
$$
\begin{lemma}\label{lem}
For any solution $(x,y)$ to (\ref{Syst:Nonau}), a finite $\rho_0=\lim_{s\rightarrow -\infty} y(s)e^{-2s}$ implies $$\lim_{s\rightarrow -\infty}\frac{x(s)}{y(s)}=\frac13\,.$$
\end{lemma} 
\begin{proof} Using de l'Hospital rule together with the system (\ref{Syst:Nonau}) one gets the claim by
$$N=\lim_{s\rightarrow -\infty}\frac{x(s)}{y(s)}=\lim_{s\rightarrow -\infty}\frac{x'(s)}{y'(s)}=\lim_{s\rightarrow -\infty}\frac{1-\frac{x(s)}{y(s)}}{2-\frac{8\pi y(s)(x(s)/y(s)+1)}{1-8\pi x(s)}}=\frac{1-N}{2}.$$
\end{proof}

\begin{figure}[ht]
\label{Fig:Sol}
\begin{center}
\includegraphics[height=5.5cm]{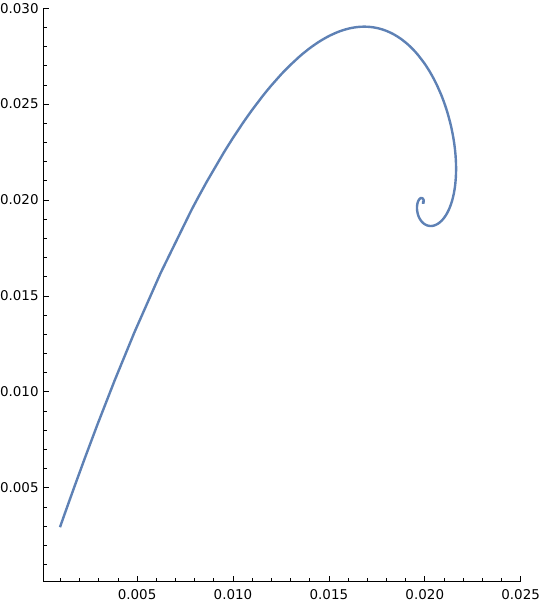}\hspace*{30pt}\includegraphics[height=5.5cm]{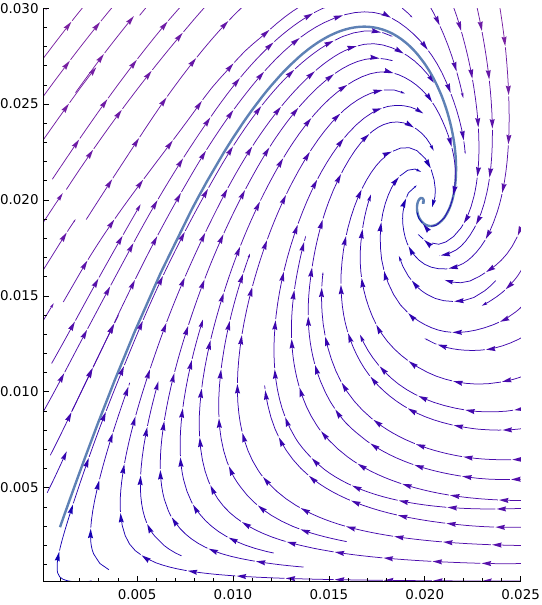}
\end{center}
\caption{\it\small Left: the heteroclinic orbit joining $(0,0)$ with $(\frac{1}{16\pi},\frac{1}{16\pi})$
 Right: the heteroclinic orbit along vector field converging to these points}
\end{figure}

Indeed the linearized version of the flow at $(0,0)$ reads
\be{Syst:Deffi}\left\{\begin{array}{l}
x'=y-x\\[6pt]
y'=2y
\end{array}\right.\ee
with the eigenvalues $-1, 2$ and the corresponding eigenvectors $(1,0)$ and $(1,3)$ resp.

The motivation for the derivation of the Lyapunov function comes from the work of Biler, Hilhorst, Nadzieja \cite{BHN} where the Newtonion case was treated. We multiply the first equation as follows
$$x'\left(-2+\frac{8\pi (x+y)}{1-8\pi x}\right)=(y-x)\left(-2+\frac{8\pi (x+y)}{1-8\pi x}\right)$$
so that after adding it with the second one $y'=y\left(2-\frac{8\pi (x+y)}{1-8\pi x}\right)$ get rid of the terms $2y$ and obtain
$$
x'\left(-2+\frac{8\pi (x+y)}{1-8\pi x}\right)+y'=-x\left(-2+\frac{8\pi (x+y)}{1-8\pi x}\right).
$$
Next our aim is to obtain the negative term on the right hand side so we use once again the second equation properly modified to the form with $C>0$ to be fixed
$$y'(y-C)/y=\left(2-\frac{8\pi (x+y)}{1-8\pi x}\right)(y-C)$$
and add to the last equality multiplied by $3$ to obtain
$$
3x'\left(-2+\frac{8\pi (x+y)}{1-8\pi x}\right)+4y'-Cy'/y=(C-y-3x)\left(-2+\frac{8\pi (x+y)}{1-8\pi x}\right)\,.
$$
Hence using $y=x+x'$ one gets
$$
3x'\left(-2+\frac{8\pi (2x+x')}{1-8\pi x}\right)+4y'-Cy'/y=\frac{(C-y-3x)(-2+24\pi x +8\pi y)}{1-8\pi x}\,.
$$
In what follows we fix the constant $C=1/{(4\pi )}$ and omit the term $x'^2$ to obtain
$$
-6x'+48\pi xx'/(1-8\pi x) +4y'-\frac{1}{4\pi}y'/y=\frac{-(1-12\pi x -4\pi y)^2}{2\pi (1-8\pi x)}\le 0\,.
$$
and after multiplication by $4\pi$ and integration using $8\pi x -1+1$ one ends up with
$$
\left(-48\pi x - 3\log (1-8\pi x)  +16\pi y -\log y \right)'\le 0
$$
whence the Lyapunov function follows
$$L(x,y)=2+16\pi (y-3x)-\log(128\pi y(1-8\pi x)^3)\,,$$
where we fixed constant so that
$$
L(1/(16\pi), 1/(16\pi))=0\,.
$$
Next we establish {\bf the critical mass theorem}.
\begin{theorem}
For the system (\ref{Syst:Nonau}) we have
$$
m/(4\pi r)=x<  9/(96\pi)=(3/4)\cdot (1/8\pi)=(3/8)\cdot (1/(16\pi ))\,
$$
or
$$
\frac{m}{r}\le \frac38\,.
$$
This corresponds to the improvement of the Buchdahl limit and Schwarzschild bound
$$
\frac{GM}{Rc^2}<3/8<4/9<1/2\,,
$$
or equivalently
$$
\frac{2GM}{Rc^2}<3/4<8/9<1\,.
$$
We could have added more estimate with our numeric estimate $0.55$ is better than that of Bondi \cite{B1, B2}, under assumption $\rho \ge 3p$, which is better than our theoretical one but with $\rho=p$ and general Bondi with nonnegative density $\rho\ge 0$ as the last one listed below
$$
\frac{2GM}{Rc^2}<0.55<0.64<0.75<0.97\,.
$$
In fact the maximum value of $2GM/Rc^2$ for a fluid described by a linear equation of state in general relativity that we have discussed above was first obtained by approximation in Chavanis papers \cite{a,b}. In particular, even better than our estimate, the numerical value $0.544$ in the case $P=rho$ appeared in Sec. 3.5 and in Fig. 12 of \cite{a} and in Secs. 2.6 and 4.2 of \cite{b} by the use of the Milne variables.

This can be rephrased as the lower limits for the radius where our estimate improves this of Buchdahl and Schwarzschild radius as the following inequalities indicate
$$
\frac{Rc^2}{GM}>2\frac23>2\frac14 >2\,.
$$
\end{theorem}
\begin{proof}
For the proof we shall use once again the Lyapunov function and the dynamical system approach.
First of all we shall show that the trajectories region bounded by the lines $y=3x$ in red, $y=-3x+1/(4\pi )$ and $y=x$ must enter finally the region bounded by the properly chosen level set for the Lyapunov function
$$
L(x,y)=1-\log (2)>0\,.
$$
Thus the maximal value of $x$ in the set $L(x,y)\le 1-\log (2)$ will yield the maximal value of the mass $m$ as stated in the theorem. 

\begin{figure}[ht]
\label{Fig:Sol}
\begin{center}
\includegraphics[height=5cm, width=5cm]{cur.pdf}\hspace*{30pt}\includegraphics[height=5.5cm, width=7cm]{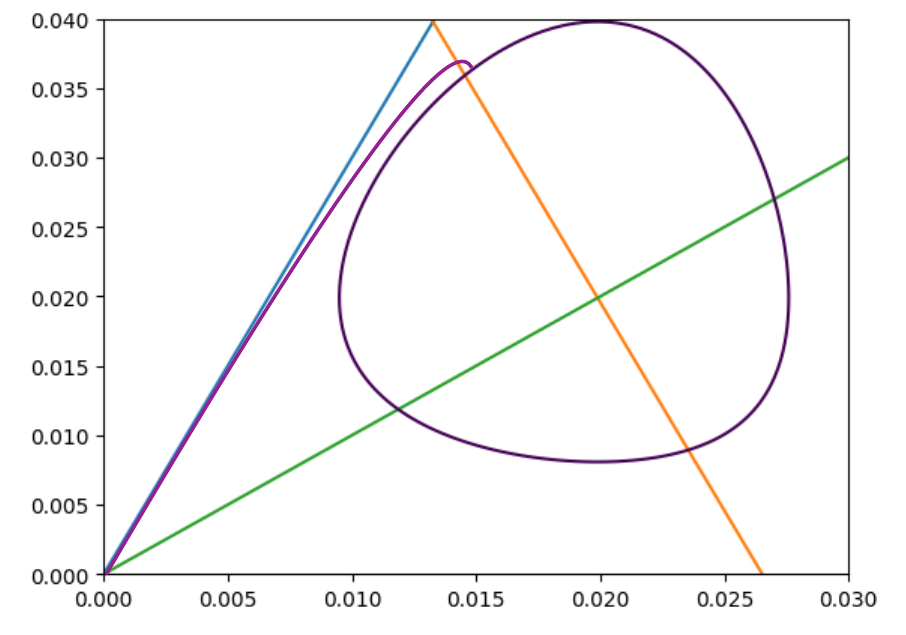}
\end{center}
\caption{\it\small Left: the heteroclinic orbit joining $(0,0)$ with $(\frac{1}{16\pi},\frac{1}{16\pi})$
 Right: level set for Lyapunov function; upper triangle with orange exit}
\end{figure}

In the sequel we shall refer to the lines and curves with special mention of color in which it appears above in the figure on the right. The green line is isocline with $x'=0$ i.e. $y=x$. The orange line is isocline with $y'=0$ i.e. $1-4\pi x - 8\pi y=0$. The blue line is the tangent line to unstable manifold, i.e. the curve of our interest emerging from $(0,0)$ at angle tangent to $y=3x$. Finally, we choose the level set of the Lyapunov function in violet color such that it is tangent to the horizontal line intersecting the point of the intersection of the two lines: the blue and the orange.

To this end notice that if $y=3x$ in blue then for $w=8\pi x$ the ratio
$$
y'/x'=3(1-4w/(1-w))\in [-3,3]
$$
if we assume $x\in [0,1/(24\pi )]$ or $w\in [0,1/3]$. This implies that all the time the vector field in (\ref{Syst:Nonau}) is directed towards interior of the triangle bounded by the three lines. Next if $y=x$ in green then $x'=0$ and $y'=\frac{2y(1-16\pi y)}{1-8\pi y}<0$ which means that up to a point $(1,1)/(16\pi )$ our vector field is directed to the left thus inwards the region bounded by the curves.
Thus there are two possibilities either our trajectory hits the boundary of $L=1-\log (2)$ in violet and enters the region $L\le 1-\log (2)$ directly with increasing $y$ or it hits the line $y=-3x+1/(4\pi )$ in orange outside the absorbing set $L\le 1-\log (2)$ in violet but then it turns down at the turning point on the line $y=-3x+1/(4\pi )$ in orange and finally enters the set $L\le 1-\log (2)$ with increasing $y$. In both cases our trajectory is trapped in the set $L\le 1-\log (2)$ in violet. Thus $x$ value of our trajectory must be bounded by the limit of the $L=1-\log (2)$ which is somewhere in the interval $(8,9)/(96\pi)$. Indeed is is easy to calculate the value of $L$ for $y=\frac{1}{16\pi}$ at $x=8/(96\pi )$ which is smaller and at  $x=9/(96\pi )$ which is bigger than $1-\log(2)$. Therefore the $x$ values of the whole our trajectory leaving $(0,0)$ at the angle adjacent to the vector $(3,1)$ mus lie on the left of $9/(96\pi)$ which ends the proof. Note that the intersection of the level set $L=1-\log(2)$ with the diagonal $y=x$ being the turning point for the trajectories in $x$ variable could potentially improve our result but alas the equation
$$2-32\pi x - \log(128\pi x(1-8\pi x)^3)=1-\log(2)$$
does not possesses a solution between $7/(96\pi)$ and $8/(96\pi)$, but between $8/(96\pi)$ and $9/(96\pi)$ which we already know.
\end{proof}

\section{Lyapunov function for a general linear equation of state}
In this section we consider the system of equations for the general linear equation of state generalizing the results from the previous section obtained for $p=\rho$ and allowing also the relations $p=\rho/3$ and $p=2\rho/3$ motivated by the asymptotics of the nonlinear equation of state discussed in the next section. Thus we consider the general equation of state in the form
\be{pkr}
p=\kappa \rho\,.
\ee
Then the system following from the TOV equation follows
\be{Syst:Gen}\left\{\begin{array}{l}
x'(s)=-\,x(s)+\,y(s)\,,\\[6pt]
y'(s)=2\,y(s)-\frac{1+\kappa}{2\kappa}\cdot \frac{y(s)(x(s)+\kappa y(s))}{1- x(s)}\,.
\end{array}\right.\ee
Note that in comparison with (\ref{Syst:Nonau}) or (\ref{Syst:General}) we multiplied both equations by $8\pi$ and the variables $x,y$ also differ by the same factor, without changing and thus abusing the notation. Now apart from $(0,0)$ we have another stationary point 
\be
((x_\kappa,y_\kappa)=\left(\frac{4\kappa}{(1+\kappa)^2+4\kappa},\frac{4\kappa}{(1+\kappa)^2+4\kappa}\right)
\ee
 corresponding to the singular density $\rho$. We shall examine the stability of these stationary points and solutions later while now we shall generalize the Lyapunov function formula obtained in the previous section to cover (\ref{pkr}) as follows
\be{V}
V=2y-(5+1/\kappa)x-2x_k \log \left(y(1-x)^{\delta_\kappa}\right)+C_\kappa
\ee 
where the exponent $\delta_\kappa$ is defined by
\be{}
8\kappa^2\delta_\kappa=(5\kappa+1)(\kappa+1)^2
\ee
while the constant $C_\kappa$ is chosen such that $V(x_\kappa,y_\kappa)=0$, to be more specific
\be{}
C_\kappa=(3+1/\kappa)x_\kappa+2x_\kappa \log \left( x_\kappa(1-x_\kappa)^{\delta_\kappa}\right)\,.
\ee
Now we shall prove directly that indeed the function $V$ defined by (\ref{V}) decreases along the trajectories of the system (\ref{Syst:Gen}). To this end we calculate for
$$
V=C+2y-\gamma x -\beta \log (y(1-x)^\delta)
$$
the derivative $V'=\frac{\partial V}{\partial x} x'+\frac{\partial V}{\partial y} y'$ equals after multiplication by $(1-x)$ to
$$-(1+\kappa)y^2+y(\delta\beta-\gamma+4+\beta(1+\kappa)/2-\gamma x^2+x(\gamma-\delta\beta+2\beta+\beta(1+\kappa)/(2\kappa)-2\beta\,.
$$
Thus it can be rephrased after fixing data $\gamma=5+1/\kappa$ and $\beta=2x_\kappa$ as
$$
(1-x)V'=-(5+1/\kappa) (x-x_\kappa)^2-(1+\kappa)(y-y_\kappa)^2
$$
showing that indeed the function $V$ is the Lyapunov function for the system (\ref{Syst:Gen}).

Next, we verify the stability of the critical points corresponding to stationary solutions for the system (\ref{Syst:Gen}). First of all we calculate the Jacobian matrix of the right hand side of (\ref{Syst:Gen}) obtaining the linearized version
\be{Syst:Gen}\left\{\begin{array}{l}
v'=-\,v+\,w\,,\\[6pt]
w'=-\frac{(1+\kappa)y}{\kappa(1-x)^2}\,v+\left(2-\frac{1+\kappa}{2\kappa}\cdot \frac{x+2\kappa y}{1- x}\right)w\,.
\end{array}\right.\ee
Therefore, at $(0,0)$ we have saddle point with the eigenvalue $-1$ corresponding to stable manifold tangent to $(1,0)$ and the eigenvalue $2$ corresponding to the unstable manifold starting from $(0,0)$ tangentially to the vector $(1,3)$ being of special interest to us as the heteroclinic orbit connecting to $(x_k,y_k)$. The latter point can be shown as stable sink at least for not too large values of $\kappa$ e.g. $\kappa\le 9$, while for larger values it is still a stable point but of different nature, e.g. a stable node. Note, that for $\kappa=1$ we have the eigenvalues at this sink equal to $-1+2i, -1-2i$.

\section{Dynamical system for general equation of state}
Note that for the general equation of state $$p=p(4\pi \rho)=p(r^{-2}y)$$ one obtains for $r=e^s$ as before
\be{Syst:General}\left\{\begin{array}{l}
x'(s)=-\,x(s)+\,y(s)\,,\\[6pt]
y'(s)=2\,y(s)-\frac{4\pi (y+r^2p(yr^{-2}))(x+r^2p(yr^{-2}))}{(1-8\pi x)p'(r^{-2}y)}\,,
\end{array}\right.\ee
Apparently if $$p=\rho$$ we end up with the aforementioned autonomous system (\ref{Syst:Nonau}) considered in the previous sections. Moreover, for the polytropic relativistic model
$$
\rho = Cp^{1/\Gamma}+\frac{p}{1-\Gamma}
$$
with $\Gamma \in (1,2)$ e.g. $\Gamma=4/3$ whence
$$
\rho = Cp^{3/4}+3p
$$
whereas for $\Gamma=7/5$ we get
$$
\rho = Cp^{5/7}+5p/2\,.
$$
Note that these examples are closely related to our Fermi--Dirac and Michie--King forms presented below, sharing some asymptorics. It is worth noting that in the \cite{Chr} remark on the limits to the Equation of State is made. Namely,
$$
p\le \rho
$$
so that speed of sound does not exceed speed of sound, for details see also Speck's PhD thesis \cite{Spe}, who repeats the argument of Christodoulou. The combination of power like and linear function was also mentioned by Heinzle, R\"ohr, Uggla \cite{HRU}. Moreover, Makino in \cite{Mak} using Poincar\'e--Bendixson theory for planar analysis of dynamical system with linear equation of state proved that the critical point corresponding to the singular density is the only point in $\omega$--limit set. Note that the lemma proved in linear case can be extended to a nonlinear one under mild assumption on the relation $\rho$ vs $p$ given by the equation of state, e.g. that it is $C^1$.
\begin{lemma}\label{lemm}
For any solution $(x,y)$ to (\ref{Syst:General}) with $C^1$ function $p$, a finite $$\rho_0=\lim_{s\rightarrow -\infty} y(s)e^{-2s}$$ implies $$\lim_{s\rightarrow -\infty}\frac{x(s)}{y(s)}=\frac13\,.$$
\end{lemma} 
\begin{proof} Using de l'Hospital rule together with the system (\ref{Syst:General}) one gets the claim with $r=e^s$ by 
$$N=\lim_{s\rightarrow -\infty}\frac{x(s)}{y(s)}=\lim_{s\rightarrow -\infty}\frac{x'(s)}{y'(s)}=\lim_{s\rightarrow -\infty}\frac{1-\frac{x(s)}{y(s)}}{2-\frac{4\pi y(s)\left(1+\frac{p(y(s)r^{-2}}{y(s)r^{-2}}\right)\left(\frac{x(s)}{y(s)}+\frac{p(y(s)r^{-2}}{y(s)r^{-2}}\right)}{(1-8\pi x(s))p'(r^{-2}y(s)}}.$$
But due to the fact that $\lim_{s\rightarrow -\infty}y(s)=0$ and that all the limits at $\rho_0$ of both $p$ and $p'$ exist we have that in fact
$$
N=\frac{1-N}{2}\,
$$
which ends the proof.
\end{proof}
Note that this lemma gives us the direction of unstable manifold where our trajectory spins off.

\section{Equation of state in relativistic Michie--King model}
Let $T$ be the temperature and $k$ be the Boltzmann constant, while $\varepsilon_c$ some threshold value for the energy at which the particles escape from the system. 
We start with the formula for the density distribution in the phase space with the vector momentum $\hat{p}$ 
$$
\rho = \frac{2m}{h^3}\int_0^\infty f_c(\hat{p})\left( 1+\frac{\varepsilon(\hat{p})}{mc^2}\right)d^3\hat{p}
$$
where the distribution $f_c$ depending on the chemical potential $\mu$ reads as
$$
f_c(\hat{p})=\left( \frac{1-\exp((\varepsilon(\hat{p}) - \varepsilon_c)/kT)}{\exp((\varepsilon(\hat{p})-\mu)/kT)+1}\right)_+
$$
while pressure denoted $p$ in this section is of the form
$$
p=\frac{4}{3h^3}\int_0^\infty f_c(\hat{p}) \varepsilon(\hat{p}) \frac{1+\varepsilon(\hat{p})/2mc^2}{1+\varepsilon(\hat{p})/mc^2}d^3\hat{p}
$$
where $\varepsilon(\hat{p})$ is the particle kinetic energy depending on the momentum $\hat{p}$
$$
\varepsilon(\hat{p})=\sqrt{c^2\hat{p}^2+m^2c^4}-mc^2\,.
$$
This leads to the pressure dependence $p(\rho)$ behaving like $\rho^{7/5}$ at zero and like $\rho/3$ at infinity, for the related formula see e.g. \cite{Sag}. Hence the simplified toy formula follows
$$
p(\rho)=\left((p_\infty\rho)^{-1}+(p_0\rho)^{-7/5}\right)^{-1}
$$
and we are left to establish the coefficients $p_\infty$ and $p_0$ governing the behaviour of the pressure $p$ at infinity and at zero, respectively, which will follow from the formulas stated above. One should note, that non-relativistic Michie--King distribution differs only in the behaviour at infinity where the exponent was $5/3$ in \cite{BSM} instead of relativistic $1$ which makes the pressure grow faster with the growth in density in the nonrelativistic case for large values of the density $\rho$, cf. \cite{BSJ}.

Due to the existence of the time--like Killing vector in the metric implies that along geodesics
$$
(\varepsilon+mc^2)/T\equiv {\rm const}=mc^2/T_R
$$ 
since at $R$ we have $\varepsilon_c=0$. Equivalently
$$
(\varepsilon_c/(mc^2)+1)T_R/T=1\,.
$$ 
See Meraffina, Ruffini, Klein and Tolman.
Next define
$$
\kappa=kT_R/(mc^2)
$$
and
$$
y=\varepsilon_c/(kT)\,.
$$
Then
$$
1-y\kappa=T_R/T\,.
$$
We introduce $x$ via the change of variables averaging over the $\hat{p}=mv/\sqrt{1-|v|^2/c^2}$
$$
xkT=\varepsilon(\hat{p})
$$
thus
$$
1+\varepsilon(\hat{p})/(mc^2)=1+x\kappa T/T_R=1+x\kappa/(1-y\kappa) 
$$
and
$$
\hat{p}^2=m^2c^2 \frac{x\kappa}{1-\kappa y}\left( \frac{x\kappa}{1-y\kappa}+2\right)
$$
whence
$$
\frac{\kappa m^2 c^2}{1-y\kappa}(x\kappa/(1-y\kappa)+1) dx=|\hat{p}|d|\hat{p}|
$$
while
$$
\rho = \frac{8\pi c^3\kappa m^4}{h^3(1-\kappa y)}\int_0^y \left( \frac{\kappa x}{1-\kappa y} +1\right)^2\left( \left(\frac{\kappa x}{1-\kappa y}+1\right)^2-1\right)^{1/2}\frac{1-\exp(x-y)}{1+\chi \exp(x-y)}dx
$$
or
$$
\rho = \frac{8\sqrt{2}\pi c^3\kappa^{3/2} m^4}{h^3(1-\kappa y)^{3/2}}\int_0^y x^{1/2}\left( \frac{\kappa x}{1-\kappa y} +1\right)^2\left( \frac{\kappa x/2}{1-\kappa y}+1\right)^{1/2}\frac{1-\exp(x-y)}{1+\chi \exp(x-y)}dx
$$
where $\chi=\exp((\varepsilon_c-\mu)/(kT))$ is constant, see Merafina--Ruffini, Klein. Now as $y\sim 0$ we have
$$
\rho(y) \sim \frac{8\sqrt{2}\pi m\kappa^{3/2}}{1+\chi} \left(\frac{ mc}{h}\right)^3\int_0^y \sqrt{x} (y-x) dx=\rho_0 y^{5/2}
$$
where 
$$\rho_0= \frac{32\sqrt{2}\pi m\kappa^{3/2}}{15(1+\chi)} \left(\frac{ mc}{h}\right)^3\,.$$ 
If $y\sim 1/\kappa$ then 
$$
\rho(y) \sim \frac{\rho_\infty}{(1-\kappa y)^\gamma} 
$$
where $\gamma=4$ and
$$\rho_\infty=8\pi c^3 m^4 h^{-3} \int_0^{1} w^3 \frac{1-\exp((w-1)/\kappa)}{1+\chi \exp((w-1)/\kappa)} dw\,.$$ 
Whereas the pressure is given by
$$
p=\frac{\eta 2\sqrt{2}\kappa}{1-\kappa y}\int_0^y \left(\frac{\kappa x}{1-\kappa y}\right)^{3/2}\left( \frac{\kappa x/2}{1-\kappa y} +1\right)^{3/2}\frac{1-\exp(x-y)}{1+\chi \exp(x-y)}dx\,,
$$
or 
$$
p=\eta \left(\frac{\kappa}{1-\kappa y}\right)^{5/2}\int_0^y x^{3/2}\left( \frac{\kappa x}{1-\kappa y} +2\right)^{3/2}\frac{1-\exp(x-y)}{1+\chi \exp(x-y)}dx\,,
$$
where
$$
\eta=\frac{8\pi m^4c^5}{3h^3}\,.
$$
Hence as $y\sim 0$ we have
$$
p(y)\sim p_0 y^{7/2}\,.
$$
where
$$
p_0=\frac{64\sqrt{2}\pi m^4c^5\kappa^{5/2}}{105(1+\chi)h^3}\,.
$$
Combining the behaviour of $\rho$ and $p$ at zero and substituting $y$ we obtain as $\rho\sim 0$
$$
p(\rho)\sim p_0 (\rho/\rho_0)^{7/5}\,.
$$
As far as the asymptotics at infinity is concerned first note that as $y\sim1/\kappa$ then
$$
p(y)\sim \frac{c^2\rho_\infty}{3(1-\kappa y)^4} 
$$
with exactly the same constant which appeared in the asymptotics of $\rho\sim\infty$ reading
$$\rho_\infty=8\pi c^3 \kappa^4 m^4 h^{-3} \int_0^{1/\kappa} x^3 \frac{1-\exp(x-1/\kappa)}{1+\chi \exp(x-1/\kappa)} dx\,.$$
Therefore, as $\rho\sim \infty$ we have
$$
p(\rho)\sim c^2\rho/3\,.
$$
One should refer these results to the papers of Chavanis \cite{e, f, a, b}. For the relativistic fermionic King model studied above the equation of state is $p\sim \rho^{7/5}$ i.e. the index of polytrope $\rho^{1+1/n}$ with $n=5/2$ for small $\rho$ (i.e. in the halo) like in the nonrelativistic case and $p=\rho_g^{4/3}$ ($\rho_g$ being the number density) i.e. $n=3$ for large $\rho$ (i.e. in the core). Furthermore, the polytropic equation of state $p=\rho_g^{4/3}$ in terms of $\rho_g$ corresponds to a linear equation of state $p=\rho/3$ in terms of $\rho$ (energy density) in accordance with the asymptotics established above. To be more specific, this was shown in \cite{a} (see in particular Sec. 3.2 and the last paragraph of the conclusion), and one can refer also to \cite{e,f, b}. We depict for $a=0.1$ and $a=0.2$ the relevant trajectory
$$
p(\rho)=\rho^{7/5}/(a+\rho^{2/5}/3)
$$

\begin{figure}[ht]
\label{Fig:rMK}
\begin{center}
\includegraphics[height=5cm, width=5.8cm]{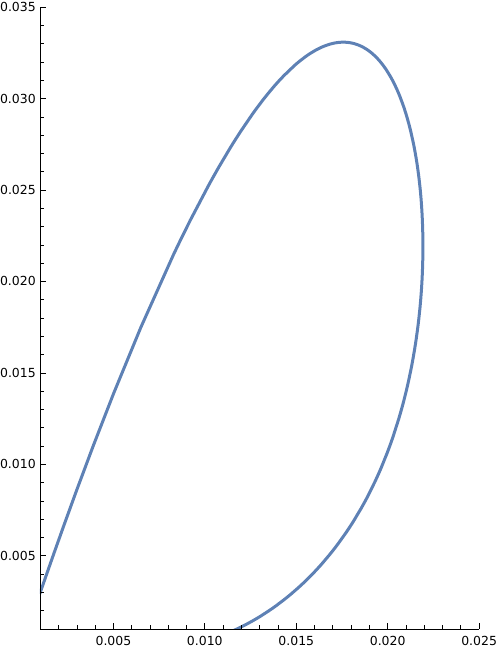}\hspace*{30pt}\includegraphics[height=5cm, width=5.8cm]{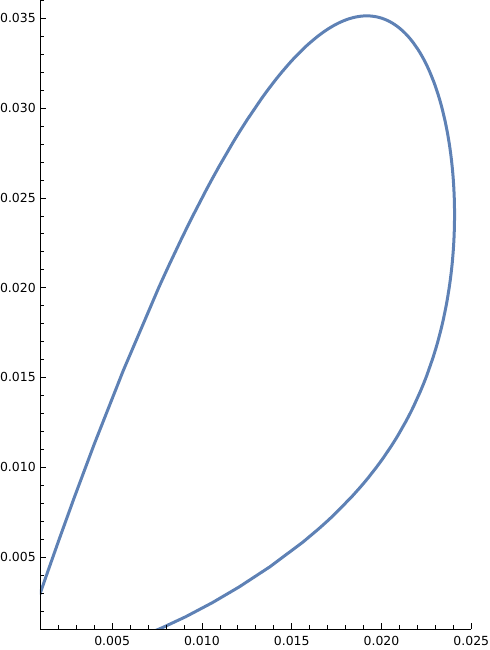}
\end{center}
\caption{\it \small The orbit from $(0,0)$ along unstable manifold for $a=0.2, 0.1$ resp.}
\end{figure}

\section{Relativistic Equation of State for Fermi--Dirac distribution}

We follow the description from Chavanis and Alberti, cf. \cite{f}. W omit the distribution function and averaging over the momenta. Thus we start with the energy density
$$
\rho = \frac{1}{\pi^2} \int_0^\infty \frac{y^{2}\sqrt{1+y^2}}{1+e^{-\alpha}e^{x\sqrt{1+y^2}}}dy
$$
and the pressure
$$
p = \frac{1}{3\pi^2} \int_0^\infty \frac{y^{4}}{\sqrt{1+y^2}\left(1+e^{-\alpha}e^{x\sqrt{1+y^2}}\right)}dy\,.
$$
At $x\sim 0$ we have after substitution $z=x\sqrt{1+y^2}$ that
$$p(x)\sim \frac{1}{3\pi x^4}\int_0^\infty z^3\left( 1+e^{-\alpha}e^z\right)^{-1}dz$$
and 
$$\rho(x) \sim \frac{1}{\pi x^4}\int_0^\infty z^3\left( 1+e^{-\alpha}e^z\right)^{-1}dz$$
whence 
$$p(\rho)\sim \rho/3$$ 
as $\rho\sim \infty$. At $x\sim \infty$ we get 
$$\pi^2\, p(x)\sim \sqrt{2}x^{-3/2}e^{-x+\alpha}\int_0^\infty w^{1/2}e^{-w}dw$$ while 
$$\pi^2\rho(x)\sim2\sqrt{2}x^{-5/2}e^{-x+\alpha}\int_0^\infty w^{1/2}e^{-w}dw.$$ Hence as $\rho\sim 0$ we get $$p(\rho)\sim -\rho/\log(\rho)\,.$$ Note that we keep the $p\le \rho$ behaviour postulated by Christodoulou so that the speed of sound does not exceed the speed of light.
The variable $x$ is the inverse of the local temperature calculated at $r$ while $\alpha>0$ is some constant related to the global temperature called also Tolman temperature measured by the observer from infinity, as well as the global chemical potential called also Klein potential and Boltzmann constant, for details see Chavanis and Alberti. It is worth noting that although the local temperature can be eliminated in the equation of state since $p=p(x)$ and $\rho=\rho(x))$ so inverting it yields the dependence of $p$ on $\rho$ directly. However, since $xp\sim \rho$ for large values of $x$ hence we can recover $x$ and hence the local temperature $1/x$ calculated at $r$ if we just take the ratio of the pressure and the energy density. Moreover, notice that in the equation of state the variable $\alpha$ is present, corresponding to the temperature in isothermal Newtonian system, yet in the limits it disappears both at $0$ and $\infty$.
$$
p(\rho)=\rho/(3-\log(\rho)(1+\rho)^{-1})
$$
\begin{figure}[ht]
\label{Fig:rFD}
\begin{center}
\includegraphics[height=5cm, width=5cm]{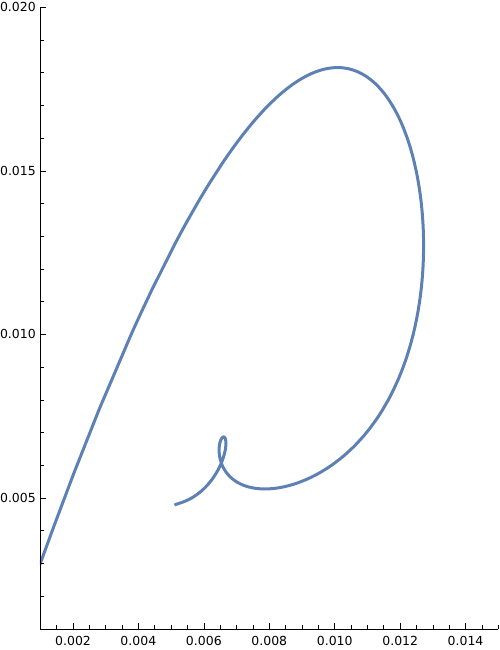}\hspace*{30pt}\includegraphics[height=5.5cm, width=6cm]{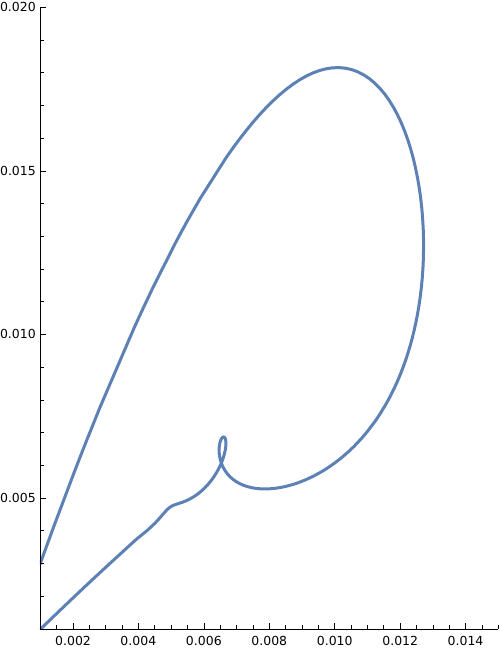}
\end{center}
\caption{\it \small The orbit starting from $(0,0)$ along unstable manifold}
\end{figure}

Above we depicted the behaviour of the branch of solution evolving alongside unstable manifold on two time intervals for the toy model mimicking the asymptotic behaviour of the pressure function in the equation of state.

Note that since the system is nonautonomous one can meet with the intersection of the trajectory on the projection plane which is not the real phase plane.

\section{Relativistic Equation of State for pure dark matter fermion gas}

We follow the description from recent paper of Arguelles, Rueda and Ruffini, cf. \cite{Arg}. W omit the distribution function and averaging over the momenta. Thus we start with the energy density
$$
\rho = \kappa \left( \sqrt{1+x^2}(x+2x^3)-{\rm arcsinh}(x)\right)
$$
and the pressure
$$
p = \frac{\kappa}{3} \left( x\sqrt{1+x^2}(2x^2-3)+3{\rm arcsinh}(x)\right) \,.
$$
Then at $x\sim 0$ we have for $8\pi^2\lambda^3\kappa=mc^2$ that
$$
\rho \sim 2\kappa x^4
$$
while
$$
3p \sim 2\kappa x^4
$$
yielding for $\rho\sim 0$
$$
3p \sim \rho\,.
$$
Moreover, as $x\sim \infty$ then
$$
3p\sim \kappa \chi x^5
$$
where $\chi=\frac{53}{40}$, while
$$
3\rho\sim 8x^3
$$
whence as $\rho\sim\infty$ we have
$$
p\sim p_\infty \rho^{5/3} 
$$
where the constant $p_\infty=1$.
The monotonicity and convexity can be deduced from
$$
\kappa^{-1}\rho'=8x^2\sqrt{x^2+1}
$$
while
$$
3\kappa^{-1}\sqrt{x^2+1}p'=8x^4
$$
and
$$
3\kappa^{-1}\sqrt{(x^2+1)^3}p''=8x^3(3x^2+4)
$$
whence by
$$
\frac{d}{d\rho}\left(p(x(\rho))\right)=\frac{p'(x(\rho)}{\rho'(x(\rho))}=\frac{1}{1+x^{-2}(\rho)}
$$
we have that the above function is positive and increasing which implies that $p$ is increasing and convex with respect to $\rho$. Furthermore
$$
3p/\rho=1-\frac{4}{2x^2+1}\le 1\,.
$$

\section{Density profiles}
Below we present density profiles for the relativistic linear case, relativistic Fermi--Dirac and relativistic Michie--King distribution functions.
\begin{figure}[ht]
\label{Fig:rFD}
\begin{center}
\includegraphics[height=3.5cm, width=3.5cm]{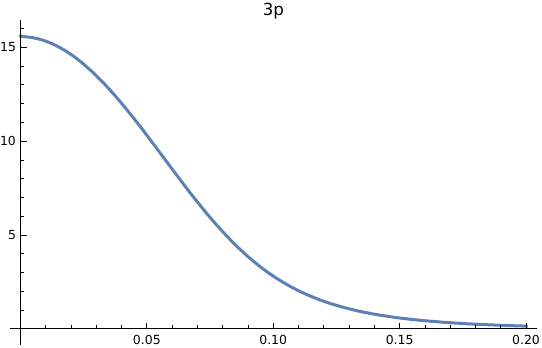}\hspace*{30pt}\includegraphics[height=3.5cm, width=3.5cm]{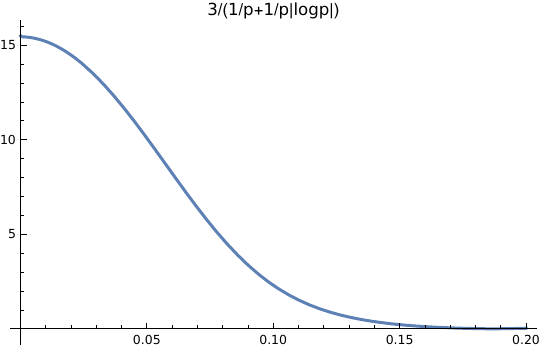}\hspace*{30pt}\includegraphics[height=3.5cm, width=3.5cm]{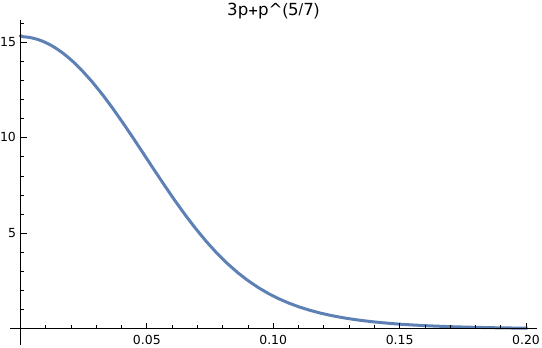}
\end{center}
\caption{\it \small The density for $\rho$ equal: $3p$, $3p/(1+1/\log p)$, $3p+p^{5/7}$, resp. }
\end{figure}
We thus depict the profiles defining the relevant equations of state involved in the solutions of the TOV equation. The rate of the convergence to zero is the only noticeable difference is in which is reasonable since the functions differ only at $0$ not at infinity where all of them should be asymptotically close to $\rho=3p$ as expected.

\section{Acknowledgements}
We dedicate this work to cherished memory of late Tadek Nadzieja, cf. \cite{BHN}, one of the pioneers in Wroc\l aw with Piotr Biler and Andrzej Krzywicki of the dynamical system approach in research on astrophysical models.

%\section{Code}
%We have used the following code for a contour of the Lyapunov function
%\begin{figure}[ht]
%\label{Fig:Lya}
%\begin{center}
%\hspace*{30pt}\includegraphics[height=1.3cm, width=11cm]{code.pdf}
%\end{center}
%\end{figure}

%\section{Table of constants}
%\begin{itemize}
%\item speed of light $c= 3\times10^{10} {\rm cm/s}$
%\item gravitational constant $G=6.67\times10^{-8} {\rm cm^3g^{-1}s^{-2}}$
%\item Planck's constant $h=4.14\times10^{-15} {\rm eV}\cdot {\rm Hz}^{-1}=eV\cdot s$
%\item Boltzmann constant $k=8.62\times 10^{-5}{\rm eV/K}$ 
%\item particle energy $mc^2=56\times 10^3 {\rm eV}$ 
%\item parsec $pc = 3\times 10^{16} {\rm m}$
%\item mass of the Sun $M=1.989\times 10^{30} {\rm kg}$
%\end{itemize}
%Thus we can calculate
%$$
%\kappa=kT_R/(mc^2)=1.977\times 10^{-5};\, \kappa^{3/2}=8.79\times 10^{-8}
%$$
%and energy cutoff parameter
%$$
%\varepsilon_c/(kT)=66.3407
%$$
%while chemical potential
%$$
%\mu/(kT)=37.7656
%$$
%thus
%$$
%\chi=\exp((\varepsilon_c-\mu)/kT)=2.57\times10^{12}\,.
%$$

%%%%%%%%%%%%%%%%%%%%%%%%%%%%%%%%%%%%%%%%%%%%%%%%%%%%%%%%%%%%%%%%%%%%%

%%%%%%%%%%%%%%%%%%%%%%%%%%%%%%%%%%%%%%%%%%%%%%%%%%%%%%%%%%%%%%%%%%%%%

\begin{thebibliography}{99}

\bibitem{f}
{\sc G. Alberti, P.H. Chavanis}, Caloric curves of
self-gravitating fermions
in general relativity, {\it Eur. Phys. J. B}, {\bf 93} (2020), 208.

\bibitem{An}
{\sc L Andersson, A.Y. Burtscher,}
On the asymptotic behaviour of static perfect fluids,
{\it Ann. Heinri Poincar\'e}, {\bf 20} (2019), 813--857.

\bibitem{Arg}
{\sc C R Argüelles,  J A Rueda,  R Ruffini},
Baryon--induced Collapse of Dark Matter Cores into Supermassive Black Holes,
{\it The Astrophysical Journal Letters}, {\bf 961} (2024), L10 6pp.

\bibitem{Sag}
{\sc E A Becerra-Vergara,  C R Argüelles,  A Krut,  J A Rueda,  R Ruffini},
Hinting a dark matter nature of Sgr A* via the S-stars,
{\it Monthly Notices of the Royal Astronomical Society}, {\bf 505} (2021), 64--68.

\bibitem{BHN}
{\sc P. Biler, D. Hilhorst, T. Nadzieja}
Existence and nonexistence of solutions for a model of garvitational interactions of particles, II,
{\it Colloq. Math}, {\bf 67} (1994), 297--308.

\bibitem{BS}
{\sc  P. Biler and R. Sta\'nczy}, 
Parabolic--elliptic systems with general density-pressure relations,
{\it Surikaisekikenkyusho Kokyuroku}, {\bf 1405} (2004), 31--53.

\bibitem{B2}
{\sc H. Bondi},
The contraction of gravitating spheres,
{\it Proceedings of the Royal Society A}, {\bf 281} (1964), 39--38.

\bibitem{B1}
{\sc H. Bondi},
Massive spheres in general relativity,
{\it Proceedings of the Royal Society A}, {\bf 282} (1964), 303--317.

\bibitem{BSA}
{\sc D. Bors, R. Stańczy,}
Existence and continuous dependence on parameters of radially symmetric solutions to astrophysical model of self-gravitating particles,
{\it Mathematical Methods in the Applied Science}, {\bf 42} (2019), 7381--7394.

\bibitem{BS}
{\sc D. Bors, R. Stańczy}, 
Models of particles of the Michie-King type,
{\it Communications in Mathematical Physics}, {\bf 382} (2021), 1243--1262.

\bibitem{BSJ}
{\sc D. Bors, R. Stańczy}, 
Dynamical system describing cloud of particles,
{\it Journal of Differential Equations,} {\bf 342} (2023), 21--23.

\bibitem{BSM}
{\sc D. Bors, R. Stańczy,} 
Mathematical model for Sagittarius A* and related Tolman-Oppenheimer-Volkoff equations,
{\it Mathematical Methods in the Applied Science}, {\bf 46} (2023), 12052--12063.

\bibitem{Buc}
{\sc H A Buchdahl},
General relativistic fluid spheres, 
{\it Phys. Rev.}, {\bf 116} (1939), 1027--1034.

\bibitem{Ch}
{\sc S. Chandrasekhar,}
A limiting case of relativistic equilibrium. In honor of J.L. Synge. in General Relativity, ed. L. O'Raifeartaigh. 1972. Oxford. Clarendon Press. 185--199.

\bibitem{a}
{\sc  P.H. Chavanis}, 
{\it Gravitational instability of finite isothermal spheres in general relativity. Analogy with neutron stars},
{\it Astronomy \& Astrophysics}, {\bf 381} (2002), 709--730.

\bibitem{b} 
{\sc P.H. Chavanis}, 
Relativistic stars with a linear equation of state: analogy with classical isothermal spheres and black holes, 
{\it Astronomy \& Astrophysics}, {\bf 483} (2008), 673--698

\bibitem{e}
{\sc P.H.  Chavanis, G. Alberti}, 
Gravitational phase
transitions and instabilities of self-gravitating fermions in general
relativity, {\it Phys. Lett. B}, {\bf 801} (2020), 135155.

\bibitem{CLM}
{\sc P.-H. Chavanis, M.~Lemou, and F. M\'ehats}
Models of dark matter halos based on statistical mechanics: The classical King model,
{\it Phys. Rev. D}, {\bf 91} (2015), 063531.

\bibitem{CSR}
{\sc P.-H. Chavanis, J.~Sommeria, and R.~Robert}, Statistical mechanics of
  two-dimensional vortices and collisionless stellar systems, {\it Astrophys. J.,} {\bf 471} (1996), 385--399.

\bibitem{Chr}
{\sc D. Christodoulou}, Self--gravitating felativistic fluids: a two--phase model, {Arch. Rational Mech. Anal.}, {\bf 130} (1995), 343--400.

\bibitem{DSD} 
{\sc J. Dolbeault and R. Sta\'nczy}, 
Bifurcation diagram and multiplicity for nonlocal elliptic equations modeling gravitating systems based on Fermi--Dirac statistics, {\it Disc. Cont. Dyn. Sys. A}, {\bf 35} (2015), 139--154.

\bibitem{DOLBEAULT:2009:HAL-00349574:2}
{\sc J.~{D}olbeault and R.~{S}ta\'nczy}, 
{N}on-existence and uniqueness results for supercritical semilinear elliptic equations, 
{\it {A}nnales {H}enri {P}oincar{\'e}}, {\bf 10} (2009), 1311--1333.


\bibitem{Ein} 
{\sc A Einstein}
On a stationary system with spherical symmetry consistring of many gravitating masses,
{\it Annals of Math.}, {\bf 40} (1939), 922--936.

\bibitem{HRU}
{\sc J.M. Heinzle, N. R\"ohr, C. Uggla},
Dynamical systems approach to relativistic symmetric static perfect fluid models,
{\it Class. Quantum Gravity}, {\bf 20} (2003), 4567--4586.

\bibitem{S5K}
{\sc S E Koposov,  D Boubert,  T S Li,  Denis Erkal,  G S Da Costa
Daniel B Zucker,  Alexander P Ji,  Kyler Kuehn,  Geraint F Lewis,  Dougal Mackey, Jeffrey D Simpson,  Nora Shipp,  Zhen Wan,  Vasily Belokurov,  Joss Bland-Hawthorn, Sarah L Martell,  Thomas Nordlander,  Andrew B Pace,  Gayandhi M De Silva},
Discovery of a nearby 1700 km/s star ejected from the Milky Way by Sgr A*
{\it Monthly Notices of the Royal Astronomical Society,} {\bf 491} (2020), 2465--2480.

\bibitem{MR0340701}
{\sc D.~D. Joseph and T.~S. Lundgren}, 
Quasilinear Dirichlet problems driven by positive sources, 
{\it Arch. Rational Mech. Anal.}, {\bf 49} (1972/73), 241--269.

\bibitem{LB}
{\sc D Lynden--Bell, R Wood}
The gravo-thermal catastrophe in isothermal spheres and the onset of red-giant structure for stellar systems,
{\it Monthly Notices of the Royal Astronomical Society,} {\bf 138} (1968), 495--525.

\bibitem{Mak}
{\sc T. Makino},
On spericallly symmetric stellar models in general relativity,
{\it J. Math. Kyoto Univ.}, {\bf 38} (1998), 55--69.

\bibitem{OS} {\sc J R Oppenheimer, H Snyder},
On continued gravitational contraction,
{\it Phys. Rev.}, {\bf 56} (1939), 455--459.

\bibitem{OV} {\sc J R Oppenheimer, G M Volkoff},
On Massive Neutron Cores,
{\it Phys. Rev.}, {\bf 55} (1939), 374--381.

\bibitem{R}
{\sc R. Robert}, 
On the gravitational collapse of stellar systems, 
{\it Class. Quantum Grav.}, {\bf 15} (1998), 3827--3840

\bibitem{RF}
{\sc  H E Russell, A C Fabian, B R McNamara, A E Broderick,}
Inside the Bondi radius of M87,
{\it Monthly Notices of the Royal Astronomical Society,} {\bf 451} (2015), 588–600.

%\bibitem{BNS}
%{\sc P.~Biler, T.~Nadzieja, and R.~Sta{\'n}czy},
%Nonisothermal systems of self-attracting {F}ermi-{D}irac particles,
%{\it Banach Center Publ.}, {\bf 66} (2004), 61--78.
 
 %\bibitem{C}
%{\sc P.-H. Chavanis}, 
%Phase transitions in self-gravitating systems,
%{\it International Journal of Modern Physics B}, {\bf 20} (2006), 3113--3198.  

\bibitem{Spe}
{\sc J.R. Speck},
On the questions of local and global well--posedness for the hyperbolic PDEs occuring in some relativistic theories of gravity and electromagnetism. Rutgers. The State University of New Jersey, Ph. D. Thesis, 2008. 

\bibitem{SDI}
{\sc R. Sta\'nczy}, 
Steady states for a system describing self-gravitating
Fermi–Dirac particles, {\it Diff. Int. Equations}, {\bf 18} (2005), 567--582.

\bibitem{SRS} 
{\sc R. Sta\'nczy}, 
The existence of equlibria of many-particle systems,
{\it Proc. Roy. Soc. Edinb. A}, {\bf 139} (2009), 623--631.

\bibitem{Tol} {\sc R C Tolman},
Static solutions of Einstein's field equations for spheres of fluid,
{\it Phys. Rev.}, {\bf 55} (1939), 364--373.


\end{thebibliography}
\end{document}